\RequirePackage{fix-cm}
\documentclass[smallextended]{svjour3}       % onecolumn (second format)
\smartqed  % flush right qed marks, e.g. at end of proof
\usepackage{graphicx}
\usepackage{amsmath}
\usepackage{amsfonts}
\usepackage{amssymb}
\usepackage{algorithmic,algorithm}
\usepackage{diagmac2}

\newcommand{\F}{\mathbb{F}}

\newcommand{\wt}{\mbox{wt}}

\newcommand{\bom}[1]{\mbox{\boldmath $#1$}}
\begin{document}

\title{Characteristic vector and weight distribution of a linear code%\thanks{Grants or other notes
}


\author{Iliya Bouyukliev    \and
        Stefka Bouyuklieva \and
        Tatsuya Maruta \and
        Paskal Piperkov
        }


\institute{Iliya Bouyukliev \at
              Institute of Mathematics and Informatics, Bulgarian Academy of Sciences, P.O.
    Box 323, Veliko Tarnovo, Bulgaria,
              \email{iliyab@math.bas.bg}           %  \\
           \and
           Stefka Bouyuklieva \at
           Faculty of Mathematics and Informatics, St. Cyril and St. Methodius University of Veliko Tarnovo, Bulgaria,
           \email{stefka@ts.uni-vt.bg}
           \and
        Tatsuya Maruta \at
        Department of Mathematical Sciences, Osaka Prefecture University,
Sakai, Osaka 599-8531, Japan,
\email{maruta@mi.s.osakafu-u.ac.jp}
        \and
        Paskal Piperkov \at
              Institute of Mathematics and Informatics, Bulgarian Academy of Sciences, P.O.
    Box 323, Veliko Tarnovo, Bulgaria,
              \email{ppiperkov@math.bas.bg}
}

\date{Received: date / Accepted: date}

\maketitle

\begin{abstract}
An algorithm for computing the weight distribution of a  linear $[n,k]$ code over a finite field $\mathbb{F}_q$ is developed. The codes are represented by their characteristic vector with respect to a given generator matrix and a generator matrix of the $k$-dimensional simplex code $\mathcal{S}_{q,k}$.
\keywords{Linear code \and Weight distribution \and Walsh transform}
\subclass{94B05 \and 15B36 \and 11Y16}
\end{abstract}

\section{Introduction}
\label{intro}

Many problems in coding theory require efficient computing of the weight distribution of a given linear code.
Some sufficient conditions for a linear code to be good or proper for error detection are expressed in terms of the weight distribution \cite{DD}. The weight distribution of the hull of a code provides a signature and the same signature computed for any permutation-equivalent code will allow the reconstruction of the permutation
\cite{Sendrier}.
The weight distributions of codes can be used to compute some characteristics of the boolean and vectorial boolean functions \cite{Carlet1}.

In 1978, Berlekamp, McEliece, and van Tilborg \cite{BMCT} proved
that two fundamental problems in coding theory, namely
maximum-likelihood decoding and computation of the
weight distribution, are NP-hard for
the class of binary linear codes.
The formal statement of the
corresponding to the weight distribution decision problem is \cite{Vardy}

\begin{tabular}{ll}
Problem:& WEIGHT DISTRIBUTION\\
Instance& A binary $m\times n$ matrix $H$ and an integer $w>0$.\\
Question:& Is there a vector $x\in\F_2^n$ of weight $w$, such that $Hx^T=0$?
\end{tabular}

Berlekamp, McEliece, and van Tilborg \cite{BMCT} proved that this problem is NP-complete.
Nevertheless, many algorithms for calculating the weight distribution have been developed. Some of them are implemented
in the software systems related to Coding theory, such as \textsc{MAGMA}, \textsc{GUAVA}, \textsc{Q-Extension}, etc.
\cite{GUAVA,MAGMA,Q-EXT}. The main idea in the common algorithms
%for finding the weight distribution of linear codes
is to obtain all linear combinations of the basis vectors and to calculate their weights. The efficient algorithms generate all codewords in a sequence, where any codeword is obtained from the
previous one by adding only one codeword. They usually use $q$-ary Gray codes (see for example \cite{GBS}) or an additional matrix (see \cite{BouyuklievBakoev}). The complexity of these algorithms is $O(nq^{k})$ for a fixed $q$. Katsman and Tsfasman in \cite{Katsman} proposed a geometric method based on algebraic-geometric codes. Some methods use matroids and Tutte polynomials, geometric lattices \cite{JP}, or Gr\"{o}bner bases \cite{Borges-Quintana2008,GS,Moro2016,Sala2007}. The algorithm in \cite{Borges-Quintana2015} is based on the idea of an ideal associated to a binary code, and its main aim is to compute the set of coset leaders, but the algorithms in that paper can be easily reformulated to compute the weight distribution of codes over different finite fields.
Bellini and Sala in \cite{BelliniSala} provided a deterministic algorithm to compute the weight and distance distribution of a binary nonlinear code, which takes advantage of fast Fourier techniques. The binary code in \cite{BelliniSala} is represented as a set of Boolean functions in numerical normal form (NNF). Efficient calculation of the weight distribution for linear codes over large finite fields is given in \cite{Han-Seo-Ju}.

%Indeed, since Gr\"{o}bner bases are involved, the
%complexity is $O(n^2 2^{n-k})$ in the binary case.
%Another approach is to consider all linear codes as a generalization of cyclic codes and use some well known ideas of Cooper for cyclic codes (see \cite{GS}).

We propose an algorithm for computing the weight distribution of a linear code based on a generalized Walsh-Hadamard transform. The linear codes here are represented by their characteristic vector $\chi$. We obtain a vector whose coordinates are all non-zero weights in the code, by multiplying a special (recursively constructed) integer matrix by $\chi^{\rm T}$. The complexity for this multiplication is $O(kq^{k}$), where $k$ is the dimension of the considered code.

%The multiplication can be realized by a butterfly algorithm which is very fast in a parallel realization. The proposed algorithm is effective especially for codes with large lengths.

 In the binary case,
 %our approach is related to the Walsh-Hadamard transform \cite{ElliottRao}, and so one can
 we compute the weight distribution by using algorithms for fast Walsh transform which are easy for implementation. For codes over prime field with $p>2$ elements we use an integer matrix of size $\theta(p,k)\times\theta(p,k)$ where $\theta(p,k)=\frac{p^k-1}{p-1}$. The weight distribution in this case can also be obtained by applying the generalized Walsh transform but then one has to use a $p^k\times p^k$ matrix \cite{Karpovsky79}. For codes over a composite field with $q=p^m$ elements, $m>1$, we use the trace map and take their images over the prime field $\mathbb{F}_p$.

 The considered algorithms  are related to butterfly networks and diagrams. A detailed description of the binary  butterfly network is presented in \cite{Parhami2002}.
These types of algorithms have very efficient natural implementations  with SIMD model of parallelization especially with the CUDA platform \cite{ButterflySorting}. The speedup of
the parallel implementation of Walsh-Hadamard transform in GPU  can be seen in \cite{Bikov_Bouyukliev}.

 We implemented the presented algorithm in a C/C++ program without special optimizations. Input data were randomly generated linear codes with lengths 300, 3000, 30000 and different dimensions over finite fields with 2, 3, 4, 5,  and 7 elements. The results of our experiments show that the  presented approach is very efficient for codes with large length and small rate.
  Such codes are useful to distinguish vectorial boolean functions up to equivalence (for different types of equivalences)
 \cite{CCZ,EdelPott}.

 %For example, calculating the weight distribution of codes with length 30000 with the presented algorithm is between 4 and 80 times faster (depending on the field) than the same calculation with Magma software system.

 %The algorithms that use listings of codewords, are faster for small lengths because they use different types of CPU optimizations (bitwise representation, AVX  instructions).

 %The algorithms that use listings of codewords, are faster for small lengths because they use bitwise representation of the codewords and bitwise operations. For larger lengths, the algorithms close to discrete Fourier transforms, can be represented by a butterfly diagram and are therefore suitable for parallelization \cite{Bikov_Bouyukliev}.

 The paper is organized as follows.
In Section \ref{sec:preliminaries} we define the basic concepts and prove important assertions that we use in the paper.
%For our purpose, we use a specially chosen generator matrix of the $k$-dimensional $q$-ary simplex code.
 In Section \ref{sec:binary} we present a butterfly algorithm for computing the weight distribution of a binary linear code. Section \ref{sec:prime} is devoted to linear codes over a prime field with $p>2$ elements. In Section \ref{sec:composite} we prove some theorems that give the connection of the weight distribution of a linear $[n,k]$ code over a composite finite field with characteristic $p$ and $q=p^m$ elements, with the weight distribution of a linear $[(q-1)n,mk]$ code over  $\F_p$.
 Section \ref{sec_res} presents the complexity of the algorithms and some experimental results.

\section{Preliminaries}
\label{sec:preliminaries}

Let $\F_q=\{0,\alpha_1=1,\alpha_2,\dots,\alpha_{q-1}\}$ be a field with $q$ elements, and $\F_q^n$ be the  $n$-dimensional vector space over $\F_q$. Every $k$-dimensional subspace $C$ of $\F_q^n$ is called a $q$-ary \textit{linear $[n,k]$
code} (or an $[n,k]_q$ code). The parameters $n$ and $k$ are called \textit{length} and
\textit{dimension} of $C$, respectively, and the vectors in $C$ are called \textit
{codewords}. The \textit{(Hamming) weight} $\wt(v)$ of a vector $v\in\F_q^n$
is the number of its non-zero coordinates. The smallest weight of a non-zero codeword is called the \textit{minimum weight} of the code. If $A_i$ is the number of
codewords of weight $i$ in $C$, $i=0,1,\dots, n$, then the sequence $(A_0, A_1,
\dots, A_n)$ is called the \textit{weight distribution} of $C$, and the polynomial $W_C(y)=\sum_{i=0}^n A_iy^i$ is the \textit{weight enumerator} of the code. Obviously, for any linear code $A_0=1$ and $A_i=0$ for $i=1,\dots,d-1$, where $d$ is the minimum weight.
%The \textit{minimum weight} $d$ of $C$ is the smallest weight among all its non-zero codewords, and so $C$ is said to be a \textit{$[n,k,d;q]$ linear code}.

Any $k\times n$ matrix $G$, whose rows form a basis of $C$, is called a \textit{generator matrix}
of the code.
The $q$-ary simplex code $\mathcal{S}_{q,k}$ is a linear code over $\F_q$ generated by a $k\times \theta(q,k)$ matrix $G_k$ having as columns a maximal set of nonproportional vectors from the vector space $\mathbb{F}_q^k$, $\theta(q,k)=(q^k-1)/(q-1)$. In other words, the columns of the matrix represent all points in the projective geometry $PG(k-1,q)$. For more information about linear codes and their parameters we refer to \cite{HP,JP,MWS}.
%All nonzero codewords of $\mathcal{S}_{q,k}$ have the same Hamming weight $q^{k-1}$, so the $q$-ary simplex code is a linear constant weight code with parameters $[\theta(q,k), k, q^{k-1}]$.

Let $C$ be a $k$-dimensional linear code over $\F_q$ and $G$ be a generator matrix of $C$. Without loss of generality we can suppose that $G$ has no zero columns (otherwise we will remove the zero columns).

\begin{definition}
The characteristic vector of the code $C$ with respect to its generator matrix $G$ is the vector
\begin{equation}\label{chi}
\chi(C,G)=\left(\chi_1,\chi_2,\ldots,\chi_{\theta(q,k)}\right)\in \mathbb{Z}^{\theta(q,k)}
\end{equation}
where $\chi_u$ is the number of the columns of $G$ that are equal or proportional to the $u$-th column of $G_k$, $u=1,\ldots,\theta(q,k)$.
\end{definition}

%The extended characteristic vector of the code $C$ with respect to its generator matrix $G$ is
%\begin{equation}
%\overline{\chi}(C,G)=(0,\underbrace{\chi|\chi|\ldots|\chi}_{q-1})
%\end{equation}
When $C$ and $G$ are clear from the context, we will write briefly $\chi$. Note that $\sum_{u=1}^{\theta(q,k)}\chi_u=n$, where $n$ is the length of $C$.

A code $C$ can have different characteristic vectors depending on the chosen generator matrices of $C$ and the considered generator matrix $G_k$ of the simplex code $\mathcal{S}_{q,k}$. If we permute the columns of the matrix $G$ we will obtain a permutation equivalent code to $C$ having the same characteristic vector. Moreover, from a characteristic vector one can restore the columns of the generator matrix $G$ but eventually at different order and/or multiplied by  nonzero elements of the field. This is not a problem for us because the equivalent codes have the same weight distributions.

All codewords of the code are the linear combinations of the rows of a given generator matrix $G$. We can easily obtain all nonzero codewords of $C$ using the multiplication
%Let we fix as basis the rows of a generator matrix $G$ of the code and $\lambda=(\lambda_1,\lambda_2,\ldots,\lambda_k)$ be a row vector from $F_q^k$.
%The linear combination of the rows of $G$ produces a codeword $\lambda.G$ where multiplication is over $F_q$.
 %Because of the propeties of the simplex code we may obtain all codewords by the following multiplication over $F_q$
\begin{equation}\label{eqCWP}
\left(\begin{array}{r}G_k^{\rm T}\\\alpha_2\, G_k^{\rm T}\\ \vdots~~ \\ \alpha_{q-1}\, G_k^{\rm T}\end{array}\right)\cdot G=
\left(\begin{array}{r}G_k^{\rm T}\cdot G\\\alpha_2\, G_k^{\rm T}\cdot G\\ \vdots~~~~~ \\ \alpha_{q-1}\, G_k^{\rm T}\cdot G\end{array}\right).
\end{equation}
To know the weight distribution of the code $C$, it is enough to compute the weights of the rows of the matrix $G^{\rm T}_k\cdot G$.

Further, we consider the matrices $M_k=G_k^{\rm T}\cdot G_k$, $k\in \mathbb N$. We denote by ${\cal N}(M_k)$ the matrix obtained from $M_k$ by replacing all nonzero elements by $1$.

\begin{lemma}\label{lemNormMult}
Let $C$ be an $[n,k]_q$ code, $G$ be its generator matrix and $\chi$ be the characteristic vector of $C$ with respect to $G$. Then the Hamming weight of the $i$-th row of the matrix $G_k^{\rm T}\cdot G$ (multiplication over $\F_q$) is equal to the $i$-th coordinate of the column vector ${\cal N}(M_k)\cdot \chi^{\rm T}$ (multiplication over $\mathbb Z$), $i=1,\ldots,\theta(q,k)$.
\end{lemma}

\begin{proof} Let $\theta=\theta(q,k)$ for short, and $s_1,\dots,s_{\theta}$ be the columns of $G_k$. Since
$M_k=(m_{ij})=G_k^T\cdot G_k$, then $m_{ij}=s_i\cdot s_j\in\F_q^k$, where $x\cdot y=x_1y_1+\cdots+x_ky_k\in\F_q$ is the Euclidean inner product of the vectors $x,y\in\F_q^k$ over $\F_q$. Similarly, $v_{ij}=s_i\cdot b_j$, where $b_1,\ldots,b_n$ are the columns of $G$, and $G_k\cdot G=(v_{ij})$. From the definition of the characteristic vector $\chi$ we know that $\chi_1$ of the columns of $G$ are proportional to $s_1$, $\chi_2$ columns are proportional to $s_2$, etc.

For $a\in\F_q$ we define ${\cal N}(a)=0$ if $a=0$ and ${\cal N}(a)=1$ otherwise. If $v_i=(v_{i1},\ldots,v_{in})$ is the $i$-th row of the matrix $G_k\cdot G$, we have
\begin{align*}
  w_i=\wt (v_i) & =\sum_{j=1}^n {\cal N}(s_i\cdot b_j)=\sum_{j=1}^n {\cal N}(s_i\cdot s_{u_j}) \\
   & =\sum_{j=1}^\theta \chi_j{\cal N}(s_i\cdot s_j)=\sum_{j=1}^\theta {\cal N}(m_{ij})\chi_j=u_i,
\end{align*}
where $u_i$ is the $i$-th coordinate of ${\cal N}(M_k)\cdot \chi^{\rm T}$.
\end{proof}

Lemma \ref{lemNormMult} and \eqref{eqCWP} show that the coordinates of the vector ${\cal N}(M_k)\cdot\chi^{\rm T}$ are all weights in a maximal set of codewords in the code $C$ with the following properties: (1) no two codewords in the set are proportional, and (2) any codeword of $C$ is proportional to a codeword belonging to the set. Hence using this matrix by vector multiplication we can obtain the weight distribution of $C$ without calculating all codewords.

To develop a fast algorithm for the proposed matrix by vector multiplication, we use a modified Walsh-Hadamard transform.
Let $h(x)=h(x_1,\dots,x_k)$ be a Boolean function in $k$ variables. Discrete
Walsh--Hadamard transform $\hat{h}$ of $h$ is the integer valued function $\hat{h}:\mathbb{F}_2^k\to\mathbb{Z}$, defined by
\begin{equation}
\hat{h}(\omega)=\sum_{x\in \mathbb{F}_2^k}h(x)(-1)^{x\cdot\omega},\quad \omega\in\mathbb{F}_2^k
\end{equation}
where $x\cdot\omega$ is the Euclidean inner product.
%$$f^W(a)=\sum_{x\in\mathbb{F}_2^n} (-1)^{f(x)\oplus \langle a,x\rangle}.$$
This transform is equivalent to the multiplication of the Truth Table of $h$ by the matrix $H_k=\otimes^k\left(\begin{array}{rr}1&1\\1&-1\end{array}\right)$.
The Kronecker power of a matrix can be represented as a product of sparse matrices \cite{Good} that leads to the more effective butterfly algorithm for calculation (fast transform).

\section{Binary codes}
\label{sec:binary}

There is a method based on the fast Walsh-Hadamard transform for the computation of the weight distribution of a given binary linear code. The complexity of this computation is $O(k2^k)$ \cite{Karpovsky79}.

In this case the columns of a generator matrix of the simplex code $\mathcal{S}_k$ are all nonzero vectors from $\F_2^k$. We take
$G_k=(\overline{1}^T \ \cdots \ \overline{2^k-1}^T)$, where $\overline{u}$ is the binary representation of the integer $u$, considered as a vector with $k$ coordinates, $1\le u\le 2^k-1$.

If $C$ is an $[n,k,d]$ binary linear code with a characteristic vector $\chi_C$, then
$$M_k\cdot \chi_C^T=(G_k^T\cdot G_k)\chi_C^T=\left(\begin{array}{c}w_1\\ \vdots \\ w_{2^k-1}\\
\end{array}\right),$$
where $w_1,\ldots,w_{2^k-1}$ are the weights of all nonzero codewords in $C$.

If
$\overline{M}_k=\left(\begin{array}{c|c}0&0\ldots 0\\ \hline \mathbf{0}^T &M_k\\
\end{array}\right)$, then the matrix $J-2\overline{M}_k$ is a square $\pm 1$ matrix of order $2^k$, which is equal to the Hadamard matrix of Sylvester type $H_k=\otimes^k\left(\begin{array}{rr}1&1\\1&-1\end{array}\right)$ (by $J$ we denote the all 1's matrix of the corresponding size).
%The matrices $H_k$ can be defined recursively as
%$$H_0=(1), \ \ H_1=\left(\begin{array}{cr} 1&1\\ 1&-1\\ \end{array}\right), \ \
%H_n=\left(\begin{array}{cr} H_{n-1}&H_{n-1}\\ H_{n-1}&-H_{n-1}\\ \end{array}\right)=H_1\otimes H_{n-1}
%\ \mbox{for} \ n\ge 2.$$
It follows that $$H_k\overline{\chi}_C^T=(J-2\overline{M}_k)\overline{\chi}_C^T=\left(\begin{array}{c}n\\ n-2w_1\\ \vdots \\ n-2w_{2^k-1}\\
\end{array}\right)=\hat{\chi}_C,$$ where $\hat{\chi}_C$ is the Walsh transform of $\overline{\chi}_C=(0,\chi_C)$, if we consider this characteristic vector as a Truth Table of a Boolean function. Hence we can obtain the weight distribution of $C$ after applying the Walsh-Hadamard transform on its characteristic vector. Algorithm \ref{Alg:butterfly1} presents the pseudo code of the corresponding butterfly implementation.

\begin{algorithm}[ht]
\caption{Butterfly Algorithm for Fast Walsh Transform}\label{Alg:butterfly1}
\begin{algorithmic}[1]
\REQUIRE The extended characteristic vector $\overline{\chi}_C$ with length $2^k$
\ENSURE An updated array $W$ -- the result of the transform
\STATE $j \leftarrow 1$; $W\leftarrow \overline{\chi}(C)$;
%\FOR {$l=1$ \TO$km$}
%\STATE $S=2^l$, $s=2^{l-1}$;
\WHILE {$j <2^k$}
\FOR {$u=0$ \TO $2^k-1$}
\IF {$\overline{u}[j]=0$}
\STATE $tt\leftarrow W[u]$;
\STATE $W[u]\leftarrow W[u]+W[u+j]$;
\STATE $W[u+j]\leftarrow tt-W[u+j]$;
\ENDIF
\ENDFOR
\STATE $j\leftarrow 2j$;
\ENDWHILE
\end{algorithmic}
\end{algorithm}

For   more details on the butterfly algorithms and diagrams we refer to \cite{Joux}. The algorithm is very suitable for parallel realization, especially with CUDA  parallel computing platform.

\section{Codes over prime fields}
\label{sec:prime}

We define a sequence of matrices $G_k$ as follows:
\begin{equation}\label{Gk}
G_1=(1),\quad G_{k}=\left(\begin{array}{ccccc}\mathbf{0}&\mathbf{1}&\ldots&\mathbf{p-1}&1\\
G_{k-1}&G_{k-1}&\ldots&G_{k-1}&\mathbf{0}^T\end{array}\right),\; k\in\mathbb{Z}, k\geq 2,
\end{equation}
where $\mathbf{u}=(u,\dots,u)=u(1,1,\dots,1)=u.$\textbf{1}, $u=0,1,\dots,p-1$. The size of the matrix $G_k$ is $k\times \theta(p,k)$. All columns in $G_k$ are pairwise linearly independent, so $G_k$ is a generator matrix of $\mathcal{S}_{p,k}$.

Let $C$ be a linear $[n,k,d]$ code over the prime field $\F_p=\{0,1,\dots,p-1\}$ with a characteristic vector $\chi$ with respect to its generator matrix $G$ and the matrix $G_k$ as defined in \eqref{Gk}. To obtain the weight distribution of $C$, we need to calculate the product $M_k\chi^T$.

 Using \eqref{Gk} we obtain a recurrence relation for the matrices $M_k$ as follows:  %$M_1=(1)$  and $\forall k\in{\mathbb Z}, k\geq 2$
\begin{equation}\label{eqMk}
M_{k}=\left(\begin{array}{ccccc}M_{k-1}&M_{k-1}&\ldots&M_{k-1}&\mbox{\textbf{0}$^T$}\\
M_{k-1}&M_{k-1}+J&\ldots&M_{k-1}+(p-1)J&\mbox{\textbf{1}$^T$}\\
M_{k-1}&M_{k-1}+2J&\ldots&M_{k-1}+2(p-1)J&\mbox{\textbf{2}$^T$}\\
\vdots\\
M_{k-1}&M_{k-1}+(p-1)J&\ldots&M_{k-1}+(p-1)^2J&\mbox{\textbf{(p-1)}$^T$}\\
\mbox{\boldmath $0$}&\mbox{\boldmath $1$}&\ldots&\mbox{\boldmath $p-1$}&1\end{array}\right),
\end{equation}
$k\in{\mathbb Z}, k\geq 2$, $M_1=(1)$.
The matrix $J$ in the above formula is the $\theta(p,k-1)\times \theta(p,k-1)$ matrix with all elements equal to 1.
The form of the matrix $G_k$ is especially chosen. It enables the possibility to have only additions of matrices in the recurrence relation \eqref{eqMk}. Denote $\theta(p,k)$ by $\theta_k$, $k\in\mathbb{N}$. Then $\theta_k=p\theta_{k-1}+1$ for $k\ge 2$.
Unfortunately, there is no comfortable recurrence relation for the matrices ${\cal N}(M_k)$. To overcome this, we introduce the matrices $M_{k}^{[\chi]}=\left(\begin{array}{c}m_1^{[\chi]}\\
\vdots\\
m_{\theta_k}^{[\chi]}\end{array}\right),
$
where $m_i^{[\chi]}=(\omega_0^{(i)},\omega_1^{(i)},\ldots,\omega_{p-1}^{(i)})\in{\mathbb Z}^p$, $m_i$ is the $i$-th row of $M_{k}$, and $\omega_u^{(i)}=\sum\{\chi_j : m_{ij}=u, 1\le j\le\theta_k\}$, $u=0,1,\dots, p-1.$

\begin{theorem}\label{thm:charspectrum} The $i$-th coordinate $w_i$ of ${\cal N}(M_k)\cdot\chi^{\rm T}$ is equal to $n-\omega_0^{(i)}$, where $n$ is the length of the code, and $m_i^{[\chi]}$ is the $i$-th row of $M_{k}^{[\chi]}$.
\end{theorem}

\begin{proof} According to the definition of $m_i^{[\chi]}$ and Lemma \ref{lemNormMult}, we have
\begin{align*}
  n-\omega_0^{(i)} & =\sum_{j=1}^{\theta_k} \chi_j - \sum\{\chi_j : m_{ij}=0, 1\le j\le\theta_k\} \\
   & =\sum\{\chi_j : m_{ij}\neq 0, 1\le j\le\theta_k\}\\
   & =\sum_{j=1}^{\theta_k} \chi_j N(m_{ij})=u_i=w_i.
\end{align*}
\end{proof}

According to Lemma \ref{lemNormMult}, the coordinates of the vector ${\cal N}(M_k)\cdot\chi^{\rm T}=(w_1,w_2,\dots,w_{\theta_k})^T$ are the weights of all codewords from a maximal subset of the code, where the maximal subset has the following properties: (1) no two codewords in the set are proportional, and (2) any codeword outside this set is proportional to a codeword belonging to the set. Hence if $N_j=\sharp \{i \ : \ w_i=j\}$, then the number of codewords of weight $j$ in the code is $A_j=(q-1)N_j$.
According to Theorem \ref{thm:charspectrum}, $w_i=n-\omega_0^{(i)}$ and so $N_j=\sharp \{i \ : \ \omega_0^{(i)}=n-j\}$.

We are looking for a recurrence relation for the matrices $M_{k}^{[\chi]}$. Our aim is to use such a relation to obtain a transform matrix which can be represented as a product of sparse matrices. This could give a butterfly algorithm for fast computation.

Using the relation $\theta(p,k)=p\cdot \theta(p,k-1)+1$, we split the characteristic vector $\chi$ of the $[n,k]_p$ code $C$ into $p+1$ parts as follows
\begin{equation}\label{eqChiSplit}
\chi=\left(\chi^{(0)}|\chi^{(1)}|\ldots|\chi^{(p-1)}|\chi^{(p)}\right)
\end{equation}
where $\chi^{(j)}\in {\mathbb Z}^{\theta(p,k-1)}$, $j=0,\ldots,p-1$, and $\chi^{(p)}\in{\mathbb Z}$.
Splitting the $i$-th row of $M_k$ similarly to \eqref{eqChiSplit}, we have
$m_i=\left(m_i^{(0)}|m_i^{(1)}|\ldots|m_i^{(p-1)}|m_i^{(p)}\right)$, $m_i^{(j)}\in \F_p^{\theta(p,k-1)}$, $j=0,\ldots,p-1$, $m_i^{(p)}\in\F_p$. According to \eqref{eqMk}, we obtain
$m_i^{(s)}[j]=a_{rj}+st$, where $i=t\theta_{k-1}+r$, $1\le r\le \theta_{k-1}$, and $a_{rj}$ are the elements of the matrix $M_{k-1}$, $1\le r,j\le\theta_{k-1}$, $0\le s\le p-1$, and $m_i^{(p)}=r$. Since $(m_i+$\textbf{1})$^{[\chi]}=\sigma(m_i^{[\chi]})$ and $(m_i+$\textbf{s})$^{[\chi]}=\sigma^s(m_i^{[\chi]})$ where where $\sigma$ is the right circular shift, $s\ge 1$, then
\begin{align*}
  m_i^{[\chi]}& =(\sum\{\chi_j : m_{ij}=u, 1\le j\le\theta(p,k)\})_{u=0}^{p-1} \\
  %& =(\sum_{s=0}^{p-1}\sum\{\chi_j : m_{ij}=u, s\theta_{k-1}+1\le j\le(s+1)\theta_{k-1}\}+r^{[\chi(p)]})_{u=0}^{p-1}\\
  &=(\sum_{s=0}^{p-1}\sum\{\chi^{(s)}_{j'} : m^{(s)}_{i}[j']=u, 1\le j'\le\theta_{k-1}\}+r^{[\chi(p)]})_{u=0}^{p-1}\\
  &=(\sum_{s=0}^{p-1}\sum\{\chi^{(s)}_{j'} : a_{rj'}+st=u, 1\le j'\le\theta_{k-1}\}+r^{[\chi(p)]})_{u=0}^{p-1}\\
  &=\sum_{s=0}^{p-1}((a_r+st)^{[\chi^{(s)}]}+r^{[\chi(p)]})_{u=0}^{p-1} =\sum_{s=0}^{p-1}(\sigma^{st}(a_r^{[\chi(s)]})+r^{[\chi(p)]})_{u=0}^{p-1}.
\end{align*}
Hence the following recurrence relation holds
%\[M_{k}^{[\chi]}=
%\left(\begin{array}{ccc}
%\sum_{s=0}^{p-1} M_{k-1}^{[\chi^{(s)}]}&+&\mbox{\textbf{0}}^{[\chi^{(p)}]}\\
%\sum_{s=0}^{p-1} (M_{k-1}+sJ)^{[\chi^{(s)}]}&+&\mbox{\textbf{1}}^{[\chi^{(p)}]}\\
%\sum_{s=0}^{p-1} (M_{k-1}+2sJ)^{[\chi^{(s)}]}&+&\mbox{\textbf{2}}^{[\chi^{(p)}]}\\
%\vdots&&\vdots\\
%\sum_{s=0}^{p-1} (M_{k-1}+(p-1)sJ)^{[\chi^{(s)}]}&+&\mbox{\textbf{(p-1)}}^{[\chi^{(p)}]}\\
%\sum_{s=0}^{p-1} \mbox{\textbf{s}}^{[\chi^{(s)}]}&+&1^{[\chi^{(p)}]}
%\end{array}\right)
%\]
\begin{equation}\label{eqChMk}
M_{k}^{[\chi]}
={\small\left(\begin{array}{ccccccc}M_{k-1}^{[\chi^{(0)}]}&+&M_{k-1}^{[\chi^{(1)}]}&+\cdots+& M_{k-1}^{[\chi^{(p-1)}]}&+&\mbox{\textbf{0}}^{[\chi^{(p)}]}\\
M_{k-1}^{[\chi^{(0)}]}&+&\sigma(M_{k-1}^{[\chi^{(1)}]})&+\cdots+&\sigma^{p-1}(M_{k-1}^{[\chi^{(p-1)}]})&+&\mbox{\textbf{1}}^{[\chi^{(p)}]}\\
M_{k-1}^{[\chi^{(0)}]}&+&\sigma^2(M_{k-1}^{[\chi^{(1)}]})&+\cdots+&\sigma^{2(p-1)}(M_{k-1}^{[\chi^{(p-1)}]})&+&\mbox{\textbf{2}}^{[\chi^{(p)}]}\\
&&&\vdots&&&\\
M_{k-1}^{[\chi^{(0)}]}&+&\sigma^{p-1}(M_{k-1}^{[\chi^{(1)}]})&+\cdots+&\sigma^{(p-1)^2}(M_{k-1}^{[\chi^{(p-1)}]})&+&\mbox{\textbf{(p-1)}}^{[\chi^{(p)}]}\\
\mbox{\textbf{0}}^{[\chi^{(0)}]}&+&\mbox{\textbf{1}}^{[\chi^{(1)}]}&+\cdots+&\mbox{\textbf{(p-1)}}^{[\chi^{(p-1)}]}& +&1^{[\chi^{(p)}]}
\end{array}\right)}
\end{equation}
 So we can use permutations and additions to compute $M_k^{[\chi]}$ from $M_{k-1}^{[\chi^{(0)}]}$, $M_{k-1}^{[\chi^{(1)}]},\ldots,M_{k-1}^{[\chi^{(p-1)}]}$ and $\chi^{(p)}$.
Moreover, \textbf{s}$^{[\chi]}$ can be obtained from \textbf{0}$^{[\chi]}$ by right circular shift operation. Note that all coordinates of \textbf{0}$^{[\chi]}$ are $0$'s except the first column whose elements are equal to the sum of all coordinates of $\chi$.

\begin{example}
If $p=3$ the recurrence relation \eqref{eqChMk} is
\begin{equation}\label{eqChMk3}
M_{k}^{[\chi]}=\left(\begin{array}{ccccccc}M_{k-1}^{[\chi^{(0)}]}&+&M_{k-1}^{[\chi^{(1)}]}&+&M_{k-1}^{[\chi^{(2)}]}&+& \mbox{\textbf{0}}^{[\chi^{(3)}]}\\
M_{k-1}^{[\chi^{(0)}]}&+&\sigma(M_{k-1}^{[\chi^{(1)}]})&+&\sigma^2(M_{k-1}^{[\chi^{(2)}]})&+&\mbox{\textbf{1}}^{[\chi^{(3)}]}\\
M_{k-1}^{[\chi^{(0)}]}&+&\sigma^2(M_{k-1}^{[\chi^{(1)}]})&+&\sigma(M_{k-1}^{[\chi^{(2)}]})&+&\mbox{\textbf{2}}^{[\chi^{(3)}]}\\
\mbox{\textbf{0}}^{[\chi^{(0)}]}&+&\mbox{\textbf{1}}^{[\chi^{(1)}]}&+&\mbox{\textbf{2}}^{[\chi^{(2)}]}&+&1^{[\chi^{(3)}]}
\end{array}\right)
\end{equation}
Let $k=3$ and $\chi=(0,4,3,2,0,8,5,1,1,4,3,2,3)$. We split $\chi$ into $4$ parts $$\chi^{(0)}=(0,4,3,2), \ \ \chi^{(1)}=(0,8,5,1), \ \ \chi^{(2)}=(1,4,3,2), \ \ \chi^{(3)}=3.$$
 %To obtain $M_3^{[\chi]}$, we have to calculate first $M_2^{[\chi^{(0)}]}$, $M_2^{[\chi^{(1)}]}$ and $M_2^{[\chi^{(2)}]}$ where
 Since
$M_2=\left(\begin{array}{llll}1&1&1&0\\1&2&0&1\\1&0&2&2\\0&1&2&1\end{array}\right)$, we have
\[M_2^{[\chi^{(0)}]}=\left(\begin{array}{rrr}2&7&0\\3&2&4\\4&0&5\\0&6&3
\end{array}\right), \ \ M_2^{[\chi^{(1)}]}= \left(\begin{array}{rrr}1&13&0\\5&1&8\\8&0&6\\0&9&5\end{array}\right) \ \ M_2^{[\chi^{(2)}]}= \left(\begin{array}{rrr}2&8&0\\3&3&4\\4&1&5\\1&6&3\end{array}\right).\]
%To explain the role of \eqref{eqChMk3} we will calculate $M_2^{[\chi^{(2)}]}$ from $M_1^{[1]}$, $M_1^{[4]}$, $M_1^{[3]}$ and $M_1^{[2]}$.
%\[M_2^{[\chi^{(2)}]}=\left(\begin{array}{llll}(0,1,0)&+(0,4,0)&+(0,3,0)&+(2,0,0)\\
%(0,1,0)&+(0,0,4)&+(3,0,0)&+(0,2,0)\\
%(0,1,0)&+(4,0,0)&+(0,0,3)&+(0,0,2)\\
%(1,0,0)&+(0,4,0)&+(0,0,3)&+(0,2,0)\end{array}\right)= \left(\begin{array}{rrr}2&8&0\\3&3&4\\4&1&5\\1&6&3\end{array}\right)\]
We calculate $M_3^{[\chi]}$ from $M_2^{[\chi^{(0)}]}$, $M_2^{[\chi^{(1)}]}$, $M_2^{[\chi^{(2)}]}$ and $\chi^{(3)}$ by \eqref{eqChMk3}:
\[M_3^{[\chi]}=\left(\begin{array}{@{}*{7}{c@{}}}
\left(\begin{array}{rrr}2&7&0\\3&2&4\\4&0&5\\0&6&3\end{array}\right)&
+&\left(\begin{array}{rrr}1&13&0\\5&1&8\\8&0&6\\0&9&5\end{array}\right)&
+&\left(\begin{array}{rrr}2&8&0\\3&3&4\\4&1&5\\1&6&3\end{array}\right)&
+&\left(\begin{array}{rrr}3&0&0\\3&0&0\\3&0&0\\3&0&0\end{array}\right)
\\
\left(\begin{array}{rrr}2&7&0\\3&2&4\\4&0&5\\0&6&3\end{array}\right)&
+&\left(\begin{array}{rrr}0&1&13\\8&5&1\\6&8&0\\5&0&9\end{array}\right)&
+&\left(\begin{array}{rrr}8&0&2\\3&4&3\\1&5&4\\6&3&1\end{array}\right)&
+&\left(\begin{array}{rrr}0&3&0\\0&3&0\\0&3&0\\0&3&0\end{array}\right)
\\
\left(\begin{array}{rrr}2&7&0\\3&2&4\\4&0&5\\0&6&3\end{array}\right)&
+&\left(\begin{array}{rrr}13&0&1\\1&8&5\\0&6&8\\9&5&0\end{array}\right)&
+&\left(\begin{array}{rrr}0&2&8\\4&3&3\\5&4&1\\3&1&6\end{array}\right)&
+&\left(\begin{array}{rrr}0&0&3\\0&0&3\\0&0&3\\0&0&3\end{array}\right)
\\
\left(\begin{array}{rrr}9&0&0\end{array}\right)&
+&\left(\begin{array}{rrr}0&14&0\end{array}\right)&
+&\left(\begin{array}{rrr}0&0&10\end{array}\right)&
+&\left(\begin{array}{rrr}0&3&0\end{array}\right)
\end{array}\right)
=\left(\begin{array}{rrr}8&28&0\\14&6&16\\19&1&16\\4&21&11\\ \hline10&11&15\\
14&14&8\\11&16&9\\11&12&13\\ \hline 15&9&12\\ 8&13&15\\9&10&17\\12&12&12\\ \hline 9&17&10\end{array}\right).\]
\end{example}

Next we define one more sequence of matrices connected to $M_k^{[\chi]}$ which we use in the algorithms.

\begin{definition}\label{defPCChS}
Let $k\in \mathbb N$ and $\chi=(\chi_1,\dots,\chi_{\theta(p,k)})\in {\mathbb Z}^{\theta(p,k)}$. The matrices $M_{k}^{[\chi]}(l)$, $l=1,\dots,k$, are defined recursively as follows
\begin{enumerate}
\item $M_k^{[\chi]}(k)=M_k^{[\chi]}$.
\item For $1\le l< k$, the vector $\chi$ is split into $p+1$ parts as in \eqref{eqChiSplit} and
\[M_k^{[\chi]}(l)=\left(\begin{array}{l}M_{k-1}^{[\chi^{(0)}]}(l)\\ M_{k-1}^{[\chi^{(1)}]}(l)\\ \ldots \\ M_{k-1}^{[\chi^{(p-1)}]}(l)\\ M_1^{[\chi^{(p)}]}\end{array}\right).\]
\end{enumerate}
\end{definition}

The matrix
 $M_k^{[\chi]}(1)$ is a $\theta_k\times p$ matrix with rows $M_1^{[\chi_i]}$,  $i=1,\dots,\theta_k$, where $\chi=(\chi_1,\dots,\chi_{\theta_k})$. Since $M_1=(1)$, the columns of the matrix $M_k^{[\chi]}(1)$ are zero vectors except the second one which is equal to $\chi$.

Note that the last row of the matrices $M_k^{[\chi]}(l)$ for $l=1,\dots,k-1$ is the same, namely $M_1^{[\chi^{(p)}]}=(0,\chi^{(p)},0,\ldots,0)$. Furthermore, the row before the last one in $M_k^{[\chi]}(l)$ is the same for $l=1,\dots,k-2$.
 Actually, for all $l<k$ there are rows equal to $M_1^{[\ast]}$ in the matrix $M_k^{[\chi]}(l)$ that are the same as in the previous matrices $M_k^{[\chi]}(l^{\prime})$, $l^{\prime}<l$. We call them \textit{inactive rows}. There are $\theta_{k-l}$ inactive rows in $M_k^{[\chi]}(l)$, $l=2,\dots,k-1$. %It is a bit difficult to calculate the positions of the inactive rows. See Algorithm~\ref{alCChS} for details.

%\begin{comment}
\begin{example}
Let $p=3$, $k=3$ and $\chi=(0,4,3,2,0,8,5,1,1,4,3,2,3)$. Then
\[M_3^{[\chi]}(1)=\left(\begin{array}{ccc}0&0&0\\0&4&0\\0&3&0\\0&2&0\\0&0&0\\0&8&0\\0&5&0\\0&1&0\\0&1&0\\0&4&0\\0&3&0\\0&2&0\\0&3&0\end{array}\right),\;
M_3^{[\chi]}(2)=\left(\begin{array}{ccc}2&7&0\\3&2&4\\4&0&5\\0&6&3\\
\hline
1&13&0\\5&1&8\\8&0&6\\0&9&5\\
\hline
2&8&0\\3&3&4\\4&1&5\\1&6&3\\
\hline
0&3&0\end{array}\right),\;
M_3^{[\chi]}(3)=\left(\begin{array}{ccc}8&28&0\\14&6&16\\19&1&16\\4&21&11\\10&11&15\\14&14&8\\11&16&9\\11&12&13\\15&9&12\\ 8&13&15\\9&10&17\\12&12&12\\ 9&17&10\end{array}\right)\]
\end{example}
%\end{comment}

%\begin{example}
%Let $q=3$, $k=4$. Using \eqref{eqChiSplit} we split a characteristic vector $\chi$ into parts as follows
%\[\chi=\big(\underbrace{\chi^{(0,0)}|\chi^{(0,1)}|\chi^{(0,2)}|\chi^{(0,3)}}_{\chi^{(0)}}|
%\underbrace{\chi^{(1,0)}|\chi^{(1,1)}|\chi^{(1,2)}|\chi^{(1,3)}}_{\chi^{(1)}}|
%\underbrace{\chi^{(2,0)}|\chi^{(2,1)}|\chi^{(2,2)}|\chi^{(2,3)}}_{\chi^{(2)}}|\chi^{(3)}
%\big)\]
%In $M_4^{[\chi]}(3)$ there is one inactive row, in $M_4^{[\chi]}(2)$ there are 4 inactive rows:
%\[M_4^{[\chi]}(3)=\left(\begin{array}{l}M_{3}^{[\chi^{(0)}]}(3)\\ M_{3}^{[\chi^{(1)}]}(3)\\ M_{3}^{[\chi^{(2)}]}(3)\\ M_1^{[\chi^{(3)}]}\end{array}\right),\quad
%M_4^{[\chi]}(2)=\left(\begin{array}{l}M_{3}^{[\chi^{(0)}]}(2)\\ M_{3}^{[\chi^{(1)}]}(2)\\ M_{3}^{[\chi^{(2)}]}(2)\\ M_1^{[\chi^{(3)}]}\end{array}\right)=
%\left(\begin{array}{l}M_{2}^{[\chi^{(0,0)}]}(2)\\ M_{2}^{[\chi^{(0,1)}]}(2)\\ M_{2}^{[\chi^{(0,2)}]}(2)\\M_1^{[\chi^{(0,3)}]}\\
%M_{2}^{[\chi^{(1,0)}]}(2)\\ M_{2}^{[\chi^{(1,1)}]}(2)\\ M_{2}^{[\chi^{(1,2)}]}(2)\\M_1^{[\chi^{(1,3)}]}\\
%M_{2}^{[\chi^{(2,0)}]}(2)\\ M_{2}^{[\chi^{(2,1)}]}(2)\\ M_{2}^{[\chi^{(2,2)}]}(2)\\M_1^{[\chi^{(2,3)}]}\\
%M_1^{[{\chi}^{(3)}]}\end{array}\right)\]
%\end{example}

Till the end of this section, we present an algorithm for calculating $M_k^{[\chi]}$ computing successively $M_k^{[\chi]}(1)$, $M_k^{[\chi]}(2)$,..., $M_k^{[\chi]}(k-1)$, $M_k^{[\chi]}(k)$. The pseudo code of the main procedure is given in Algorithm \ref{alCChS}.

\begin{algorithm}[ht]
\caption{Main Procedure}\label{alCChS}
\begin{algorithmic}[1]
\REQUIRE a prime $p$, an integer $k$, and a vector $\chi$ of length $\theta(p,k)$ with integer coordinates
\quad\COMMENT{$k$ is the dimension of the considered $p$-ary code given by its characteristic vector $\chi$}
%\REQUIRE $\theta(q,k)\times q$ table $CH$\quad\COMMENT{$M_k^{[\chi]}(i-1)$}
%\REQUIRE an integer $i$\quad\COMMENT{$2\leq i\leq k$}
%\REQUIRE a small step size $stepsize$\quad\COMMENT{$\theta(q,i-1)$}
%\REQUIRE a big step size $newstepsize$\quad\COMMENT{$\theta(q,i)$}
\ENSURE the array $H$\quad\COMMENT{$H=M_k^{[\chi]}$}
\STATE $H:=M_k^{[\chi]}(1)$
%\STATE $\theta:=\theta(q,k)=\frac{q^{k}-1}{q-1}$;
\STATE $\theta_1:=1$;
\FOR {$l=2$ \TO$k$}
%\STATE $u=0$\quad\COMMENT{an index for rows of $CH$}
\STATE Initialize an array $a$ of length $k$, $a:=\textbf{0}$\quad\COMMENT{a help array for monitoring the inactive rows}
\STATE $\theta_0:=\theta_1$;
\STATE $\theta_1:=\theta(p,l)=\frac{p^{l}-1}{p-1}=p\theta_0+1$;
\STATE $r:=0$;
\WHILE {$r<\theta$}
\STATE $r_0:=r$ \quad\COMMENT{$r_0+1$ is the index of the first row of the considered submatrix}
\STATE $r:=r+\theta_1$\quad\COMMENT{the index for the last row of the considered submatrix}
\STATE NewH($H,r_0,r,\theta_0$) \quad\COMMENT{Computes $M_{l}^{[\chi^{(*)}]}$  for the current part of $\chi$}
\STATE $s:=l$
\STATE $a[s]:=a[s]+1$
\WHILE {$a[s]=q$}
\STATE $r:=r+1$\quad\COMMENT{skipping an inactive row}
\STATE $a[s]:=0$
\STATE $s:=s+1$
\STATE $a[s]:=a[s]+1$
\ENDWHILE
\ENDWHILE
\ENDFOR
\end{algorithmic}
\end{algorithm}

Algorithm \ref{alBChS} shows how to obtain $M_k^{[\chi]}(l)$ from $M_k^{[\chi]}(l-1)$. It consists of three main transformations which we call \textsc{Add0}, \textsc{LastRow} and \textsc{AllRows}.
Let explain them in the case $l=k$.
We start with the array

$$M_{k}^{[\chi]}(k-1)=\left(\begin{array}{l}M_{k-1}^{[\chi^{(0)}]}\\ M_{k-1}^{[\chi^{(1)}]}\\ \ldots \\ M_{k-1}^{[\chi^{(p-1)}]}\\ M_1^{[\chi^{(p)}]}\end{array}\right)=\left(\begin{array}{c}M_{k-1}^{[\chi^{(0)}]}\\ M_{k-1}^{[\chi^{(1)}]}\\ \ldots \\ M_{k-1}^{[\chi^{(p-1)}]}\\ 0,\chi^{(p)},0,\dots,0\end{array}\right).$$

%Note that $a_q$ take part in \eqref{eqFt} and \eqref{eqFq} as the element in the zero position of $a_1$.

%%%%%%%%%%%%%%%%%%%%%%%%%%%%%%%%%%%%%%%%%%%%%%%%%
\begin{enumerate}
\item \textsc{Add0}: First we apply the left circular shift operation on the last row of the matrix $M_k^{[\chi]}(k-1)$. Then we add the obtained vector $\mbox{lcs}(M_1^{[\chi^{(p)}]})=(\chi^{(p)},0,\dots,0)$ to all rows of $M_{k-1}^{[\chi^{(1)}]}$.
$$M_{k}^{[\chi]}(k-1)=\left(\begin{array}{c}M_{k-1}^{[\chi^{(0)}]}\\ M_{k-1}^{[\chi^{(1)}]}\\ \ldots \\ M_{k-1}^{[\chi^{(p-1)}]}\\ 0,\chi^{(p)},0,\dots,0\end{array}\right)\longrightarrow
\left(\begin{array}{c}M_{k-1}^{[\chi^{(0)}]}\\ M_{k-1}^{[\chi^{(1)}]}+\bom{0}^{[\chi^{(p)}]}\\ \ldots \\ M_{k-1}^{[\chi^{(p-1)}]}\\ 0,\chi^{(p)},0,\dots,0\end{array}\right)$$

\item \textsc{LastRow}: In this step we calculate the last row of $M_k^{[\chi]}(k)$ equal to
\begin{align*}
  \sum_{s=0}^{p-1}\mbox{\boldmath $s$}^{[\chi^{(s)}]}+1^{[\chi^{(p)}]} & = (\sum_{i=1}^{\theta_{k-1}}\chi_i,\sum_{i=\theta_{k-1}+1}^{2\theta_{k-1}}\chi_i+\chi^{(p)},\ldots, \sum_{i=\theta_k-\theta_{k-1}}^{\theta_k-1}\chi_i) \\
   & =(\sum_{i=0}^{p-1}\omega_{0,i},\sum_{i=0}^{p-1}\omega_{1,i},\ldots,\sum_{i=0}^{p-1}\omega_{p-1,i}).
\end{align*}
where $(\omega_{j,0},\omega_{j,1},\dots,\omega_{j,p-1})$ is the first row of the matrix $M_{k-1}^{[\chi^{(j)}]}$, $j=0,2,\dots,p-1$, and $(\omega_{1,0},\dots,\omega_{1,p-1})$ is the first row of the transformed in \textsc{Add0} submatrix $M_{k-1}^{[\chi^{(1)}]}$.

\begin{algorithm}[ht]
\caption{Function NewH($H,r_0,r,\theta$) }\label{alBChS}
\begin{algorithmic}[1]
\REQUIRE The array $H$ and the integers $r_0,r,\theta$ \quad\COMMENT{parameters that fix a considered submatrix}
%\REQUIRE $\theta(q,k)\times q$ table $CH$\quad\COMMENT{with values of $M_k^{[\chi]}(i-1)$ throughout the part}
%\REQUIRE a position $su$ of the first row of the part
%\REQUIRE a position $u$ of the last row of the part
%\REQUIRE a small step size $stepsize$\quad\COMMENT{$\theta(q,i-1)$}
\ENSURE an updated array $H$\quad\COMMENT{in range of the considered submatrix}
\STATE Initialize the auxiliary array $T$ of size $p\times p$
%\STATE $TEMP[i]:=CH[i+rb-1]$, $i=1,\dots,\quad\COMMENT{The last row of the part is unique for all tuples.}
%\STATE Initialize the row $CH[u]$ to $0$
\FOR {$i=1$ \TO$\theta$}
\STATE $H[r_0+\theta+i]:=H[r_0+\theta+i]+\mbox{lcs}(H[r])$ \quad\COMMENT{The transformation \textsc{Add0}}
\ENDFOR
\STATE $H[r]=(\displaystyle\sum_{i=0}^{p-1}H[r_0+1,i],\sum_{i=0}^{p-1}H[r_0+\theta+1,i],\ldots,\sum_{i=0}^{p-1}H[r_0+(p-1)\theta+1,i])$; \COMMENT{\textsc{LastRow}}
\FOR {$i=1$ \TO$\theta$}
\FOR {$j=0$ \TO$p-1$}
\STATE $T[j]:=H[r_0+j\cdot\theta+i]$
\ENDFOR
\STATE $H[r_0+i]:=T[0]+T[1]+\cdots+T[p-1]$
\FOR {$j=1$ \TO$q-1$}
\STATE $H[r_0+j\cdot\theta+i]:=T[0]+\sigma^j(T[1])+\cdots+\sigma^{j(p-1)}(T[p-1])$
\quad\COMMENT{\textsc{AllRows}}
\ENDFOR
\ENDFOR
\end{algorithmic}
\end{algorithm}

\item \textsc{AllRows}: This transformation consists of $p$ similar steps \textsc{AllRows[j]}, $j=0,1,\dots,p-1$, repeated $\theta_{k-1}$ times. To realize this transformation, we use an auxiliary $p\times p$ array $T$.  \textsc{AllRows[j]} acts on $T$ as follows:\\
   \textsc{AllRows[0]($T$)}  =$T$[0]+$T$[1] $+\cdots+$ $T$[$p-1$],\\
      \textsc{AllRows[j]($T$)}  =$T$[0]+$\sigma^j$($T$[1]) $+\cdots+$ $\sigma^{j(p-1)}$($T$[$p-1$]) for $j>0$.

    In the beginning $T$ consists of the first rows of all submatrices $M_{k-1}^{[\chi^{(j)}]}$, and in the $i$-th step $T$ consists of the $i$-th rows of these submatrices. Hence the transformation \textsc{AllRows} gives us
 $${\small
\left(\begin{array}{ccccccc}M_{k-1}^{[\chi^{(0)}]}&+&M_{k-1}^{[\chi^{(1)}]}&+\cdots+& M_{k-1}^{[\chi^{(p-1)}]}&+&\mbox{\boldmath $0$}^{[\chi^{(p)}]}\\
M_{k-1}^{[\chi^{(0)}]}&+&\sigma(M_{k-1}^{[\chi^{(1)}]})&+\cdots+&\sigma^{p-1}(M_{k-1}^{[\chi^{(p-1)}]})&+&\mbox{\boldmath $1$}^{[\chi^{(p)}]}\\
M_{k-1}^{[\chi^{(0)}]}&+&\sigma^2(M_{k-1}^{[\chi^{(1)}]})&+\cdots+&\sigma^{2(p-1)}(M_{k-1}^{[\chi^{(p-1)}]})&+&\mbox{\boldmath $2$}^{[\chi^{(p)}]}\\
&&&\ddots&&&\\
M_{k-1}^{[\chi^{(0)}]}&+&\sigma^{p-1}(M_{k-1}^{[\chi^{(1)}]})&+\cdots+&\sigma^{(p-1)^2}(M_{k-1}^{[\chi^{(q-1)}]})& +&\mbox{\boldmath $(p-1)$}^{[\chi^{(p)}]}\\
\mbox{\boldmath $0$}^{[\chi^{(0)}]}&+&\mbox{\boldmath $1$}^{[\chi^{(1)}]}&+\cdots+&\mbox{\boldmath $(p-1)$}^{[\chi^{(p-1)}]}&+&1^{[\chi^{(p)}]}
\end{array}\right)}=M_k^{[\chi]}.$$
\end{enumerate}

We keep the inactive rows unchanged in the computation of $M_k^{[\chi]}(l)$ from  $M_k^{[\chi]}(l-1)$, and apply the transformations described above to obtain $M_l^{[\chi^{\prime}]}(l)$ from $M_l^{[\chi^{\prime}]}(l-1)$ where $\chi^{\prime}$ is a suitable part of $\chi$.

\begin{example} Let $q=3$, $k=3$, and $\chi=(0,4,3,2,0,8,5,1,1,4,3,2,3)$. Applying Algorithms \ref{alCChS}--\ref{alBChS}  we have

%when $l=2$
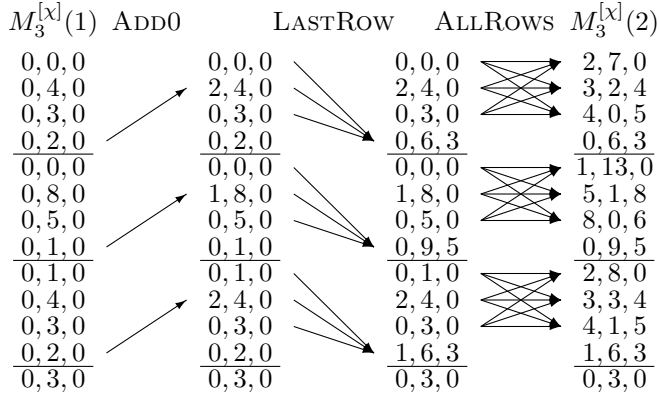
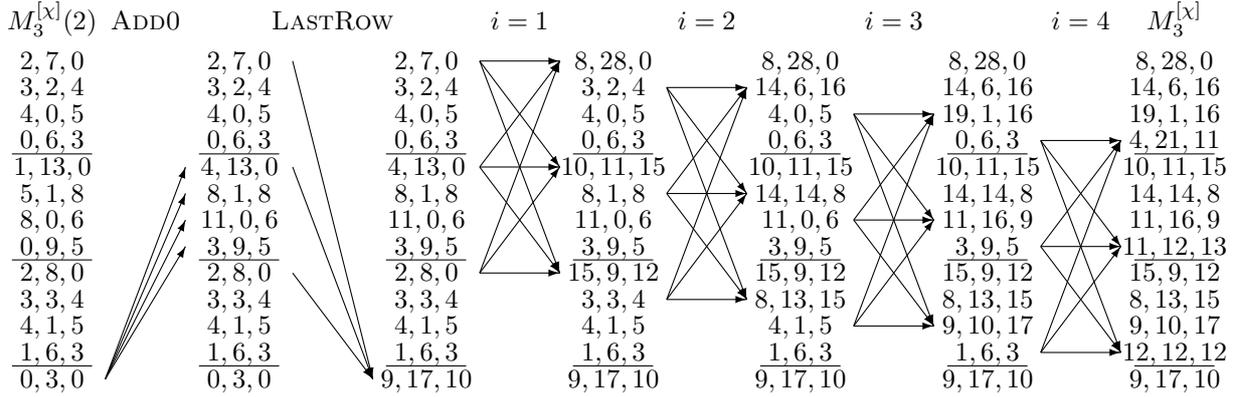
\[\ctdiagram{
\ctv -175,145:{\mbox{\textsc{Add0}}}
\ctv -105,145:{\mbox{\textsc{LastRow}}}
\ctv -45,145:{\mbox{\textsc{AllRows}}}
\ctv -210,145:{M_3^{[\chi]}(1)}
\ctv -210,130:{0,0,0}
\ctv -210,120:{0,4,0}
\ctv -210,110:{0,3,0}
\ctv -210,100:{0,2,0}
\cthead\cten -190,100,-160,120:
\ctnohead\cten -225,95,-195,95:
\ctv -210,90:{0,0,0}
\ctv -210,80:{0,8,0}
\ctv -210,70:{0,5,0}
\ctv -210,60:{0,1,0}
\cthead\cten -190,60,-160,80:
\ctnohead\cten -225,55,-195,55:
\ctv -210,50:{0,1,0}
\ctv -210,40:{0,4,0}
\ctv -210,30:{0,3,0}
\ctv -210,20:{0,2,0}
\cthead\cten -190,20,-160,40:
\ctnohead\cten -225,15,-195,15:
\ctv -210,10:{0,3,0}
\ctv -140,130:{0,0,0}
\cthead\cten -120,130,-90,100:
\ctv -140,120:{2,4,0}
\cthead\cten -120,120,-90,100:
\ctv -140,110:{0,3,0}
\cthead\cten -120,110,-90,100:
\ctv -140,100:{0,2,0}
\ctnohead\cten -155,95,-125,95:
\ctv -140,90:{0,0,0}
\cthead\cten -120,90,-90,60:
\ctv -140,80:{1,8,0}
\cthead\cten -120,80,-90,60:
\ctv -140,70:{0,5,0}
\cthead\cten -120,70,-90,60:
\ctv -140,60:{0,1,0}
\ctnohead\cten -155,55,-125,55:
\ctv -140,50:{0,1,0}
\cthead\cten -120,50,-90,20:
\ctv -140,40:{2,4,0}
\cthead\cten -120,40,-90,20:
\ctv -140,30:{0,3,0}
\cthead\cten -120,30,-90,20:
\ctv -140,20:{0,2,0}
\ctnohead\cten -155,15,-125,15:
\ctv -140,10:{0,3,0}
\ctv -70,130:{0,0,0}
\cthead\cten -50,130,-20,130:
\cthead\cten -50,130,-20,120:
\cthead\cten -50,130,-20,110:
\ctv -70,120:{2,4,0}
\cthead\cten -50,120,-20,130:
\cthead\cten -50,120,-20,120:
\cthead\cten -50,120,-20,110:
\ctv -70,110:{0,3,0}
\cthead\cten -50,110,-20,130:
\cthead\cten -50,110,-20,120:
\cthead\cten -50,110,-20,110:
\ctv -70,100:{0,6,3}
\ctnohead\cten -85,95,-55,95:
\ctv -70,90:{0,0,0}
\cthead\cten -50,90,-20,90:
\cthead\cten -50,90,-20,80:
\cthead\cten -50,90,-20,70:
\ctv -70,80:{1,8,0}
\cthead\cten -50,80,-20,90:
\cthead\cten -50,80,-20,80:
\cthead\cten -50,80,-20,70:
\ctv -70,70:{0,5,0}
\cthead\cten -50,70,-20,90:
\cthead\cten -50,70,-20,80:
\cthead\cten -50,70,-20,70:
\ctv -70,60:{0,9,5}
\ctnohead\cten -85,55,-55,55:
\ctv -70,50:{0,1,0}
\cthead\cten -50,50,-20,50:
\cthead\cten -50,50,-20,40:
\cthead\cten -50,50,-20,30:
\ctv -70,40:{2,4,0}
\cthead\cten -50,40,-20,50:
\cthead\cten -50,40,-20,40:
\cthead\cten -50,40,-20,30:
\ctv -70,30:{0,3,0}
\cthead\cten -50,30,-20,50:
\cthead\cten -50,30,-20,40:
\cthead\cten -50,30,-20,30:
\ctv -70,20:{1,6,3}
\ctnohead\cten -85,15,-55,15:
%\ctv -105,145:{cor=1}
%\ctv -35,145:{cor=1}
\ctv 0,145:{M_3^{[\chi]}(2)}
\ctv 0,130:{2,7,0}
\ctv 0,120:{3,2,4}
\ctv 0,110:{4,0,5}
\ctv 0,100:{0,6,3}
\ctnohead\cten -15,95,15,95:
\ctv 0,90:{1,13,0}
\ctv 0,80:{5,1,8}
\ctv 0,70:{8,0,6}
\ctv 0,60:{0,9,5}
\ctnohead\cten -15,55,15,55:
\ctv 0,50:{2,8,0}
\ctv 0,40:{3,3,4}
\ctv 0,30:{4,1,5}
\ctv 0,20:{1,6,3}
\ctnohead\cten -15,15,15,15:
\ctv -70,10:{0,3,0}
\ctv 0,10:{0,3,0}
}\]

%when $l=3$
\[\ctdiagram{
\ctv -35,145:{\mbox{\textsc{Add0}}}
\ctv 35,145:{\mbox{\textsc{LastRow}}}
\ctv -70,145:{M_3^{[\chi]}(2)}
\ctv 105,145:{i=1}
%\ctv 35,145:{cor=1}
\ctv 175,145:{i=2}
\ctv 245,145:{i=3}
\ctv 315,145:{i=4}
\ctv 350,145:{M_3^{[\chi]}}
\ctv -70,130:{2,7,0}
\ctv -70,120:{3,2,4}
\ctv -70,110:{4,0,5}
\ctv -70,100:{0,6,3}
\ctnohead\cten -85,95,-55,95:
\ctv -70,90:{1,13,0}
\ctv -70,80:{5,1,8}
\ctv -70,70:{8,0,6}
\ctv -70,60:{0,9,5}
\ctnohead\cten -85,55,-55,55:
\ctv -70,50:{2,8,0}
\ctv -70,40:{3,3,4}
\ctv -70,30:{4,1,5}
\ctv -70,20:{1,6,3}
\ctnohead\cten -85,15,-55,15:
\ctv -70,10:{0,3,0}
\ctv 0,130:{2,7,0}
\ctv 0,120:{3,2,4}
\ctv 0,110:{4,0,5}
\ctv 0,100:{0,6,3}
\ctnohead\cten -15,95,15,95:
\ctv 0,90:{4,13,0}
\ctv 0,80:{8,1,8}
\ctv 0,70:{11,0,6}
\ctv 0,60:{3,9,5}
\ctnohead\cten -15,55,15,55:
\ctv 0,50:{2,8,0}
\ctv 0,40:{3,3,4}
\ctv 0,30:{4,1,5}
\ctv 0,20:{1,6,3}
\ctnohead\cten -15,15,15,15:
\ctv 0,10:{0,3,0}
\ctv 70,130:{2,7,0}
\ctv 70,120:{3,2,4}
\ctv 70,110:{4,0,5}
\ctv 70,100:{0,6,3}
\ctnohead\cten 55,95,85,95:
\ctv 70,90:{4,13,0}
\ctv 70,80:{8,1,8}
\ctv 70,70:{11,0,6}
\ctv 70,60:{3,9,5}
\ctnohead\cten 55,55,85,55:
\ctv 70,50:{2,8,0}
\ctv 70,40:{3,3,4}
\ctv 70,30:{4,1,5}
\ctv 70,20:{1,6,3}
\ctnohead\cten 55,15,85,15:
\ctv 70,10:{9,17,10}
\cthead\cten -50,10,-20,90:
\cthead\cten -50,10,-20,80:
\cthead\cten -50,10,-20,70:
\cthead\cten -50,10,-20,60:
\ctv 140,130:{8,28,0}
\ctv 140,120:{3,2,4}
\ctv 140,110:{4,0,5}
\ctv 140,100:{0,6,3}
\ctnohead\cten 125,95,155,95:
\ctv 140,90:{10,11,15}
\ctv 140,80:{8,1,8}
\ctv 140,70:{11,0,6}
\ctv 140,60:{3,9,5}
\ctnohead\cten 125,55,155,55:
\ctv 140,50:{15,9,12}
\ctv 140,40:{3,3,4}
\ctv 140,30:{4,1,5}
\ctv 140,20:{1,6,3}
\ctnohead\cten 125,15,155,15:
\ctv 140,10:{9,17,10}
\cthead\cten 90,130,120,130:
\cthead\cten 90,90,120,130:
\cthead\cten 90,50,120,130:
\cthead\cten 90,130,120,90:
\cthead\cten 90,90,120,90:
\cthead\cten 90,50,120,90:
\cthead\cten 90,130,120,50:
\cthead\cten 90,90,120,50:
\cthead\cten 90,50,120,50:
\cthead\cten 20,130,50,10:
\cthead\cten 20,90,50,10:
\cthead\cten 20,50,50,10:
\ctv 210,130:{8,28,0}
\ctv 210,120:{14,6,16}
\ctv 210,110:{4,0,5}
\ctv 210,100:{0,6,3}
\ctnohead\cten 195,95,225,95:
\ctv 210,90:{10,11,15}
\ctv 210,80:{14,14,8}
\ctv 210,70:{11,0,6}
\ctv 210,60:{3,9,5}
\ctnohead\cten 195,55,225,55:
\ctv 210,50:{15,9,12}
\ctv 210,40:{8,13,15}
\ctv 210,30:{4,1,5}
\ctv 210,20:{1,6,3}
\ctnohead\cten 195,15,225,15:
\ctv 210,10:{9,17,10}
\cthead\cten 160,120,190,120:
\cthead\cten 160,80,190,120:
\cthead\cten 160,40,190,120:
\cthead\cten 160,120,190,80:
\cthead\cten 160,80,190,80:
\cthead\cten 160,40,190,80:
\cthead\cten 160,120,190,40:
\cthead\cten 160,80,190,40:
\cthead\cten 160,40,190,40:
\ctv 280,130:{8,28,0}
\ctv 280,120:{14,6,16}
\ctv 280,110:{19,1,16}
\ctv 280,100:{0,6,3}
\ctnohead\cten 265,95,295,95:
\ctv 280,90:{10,11,15}
\ctv 280,80:{14,14,8}
\ctv 280,70:{11,16,9}
\ctv 280,60:{3,9,5}
\ctnohead\cten 265,55,295,55:
\ctv 280,50:{15,9,12}
\ctv 280,40:{8,13,15}
\ctv 280,30:{9,10,17}
\ctv 280,20:{1,6,3}
\ctnohead\cten 265,15,295,15:
\ctv 280,10:{9,17,10}
\cthead\cten 230,110,260,110:
\cthead\cten 230,70,260,110:
\cthead\cten 230,30,260,110:
\cthead\cten 230,110,260,70:
\cthead\cten 230,70,260,70:
\cthead\cten 230,30,260,70:
\cthead\cten 230,110,260,30:
\cthead\cten 230,70,260,30:
\cthead\cten 230,30,260,30:
\ctv 350,130:{8,28,0}
\ctv 350,120:{14,6,16}
\ctv 350,110:{19,1,16}
\ctv 350,100:{4,21,11}
\ctnohead\cten 335,95,365,95:
\ctv 350,90:{10,11,15}
\ctv 350,80:{14,14,8}
\ctv 350,70:{11,16,9}
\ctv 350,60:{11,12,13}
\ctnohead\cten 335,55,365,55:
\ctv 350,50:{15,9,12}
\ctv 350,40:{8,13,15}
\ctv 350,30:{9,10,17}
\ctv 350,20:{12,12,12}
\ctnohead\cten 335,15,365,15:
\ctv 350,10:{9,17,10}
\cthead\cten 300,100,330,100:
\cthead\cten 300,60,330,100:
\cthead\cten 300,20,330,100:
\cthead\cten 300,100,330,60:
\cthead\cten 300,60,330,60:
\cthead\cten 300,20,330,60:
\cthead\cten 300,100,330,20:
\cthead\cten 300,60,330,20:
\cthead\cten 300,20,330,20:
}\]
\end{example}

To explain more formally the main algorithm we introduce a matrix representation of the transform steps.
We put all rows of $M_k^{[\chi]}(l)$  in one row vector of length $p\theta(p,k)$ denoted by $\widehat{M}_k^{[\chi]}(l)$, $l=1,\dots,k$. We denote $\widehat{M}_k^{[\chi]}=\widehat{M}_k^{[\chi]}(k)$ and $\widehat{\chi}=\widehat{M}_k^{[\chi]}(1)$ for short.

In the following theorem, we use matrices of three types, namely:
\begin{itemize}
\item the $p\times p$ permutation matrix $P=\displaystyle\left(\begin{array}{cc}\bom{0}&1\\ I_{p-1}&\bom{0}^T\end{array}\right)$ which realizes the circular shift right operation. Then $P^0=I_p$, and $P^j$ realizes the circular shift right operation by $j$ positions;
\item the $p\times p$ matrices $E_j$, $j=0,1,\dots,p-1$, where the $j+1$-th row of $E_j$ is the all-ones vector, and the other rows of the matrix are zero vectors;
\item  the matrices $T_{k,l}$ for  $k,l\in\mathbb Z$, $2\leq l\leq k$. We define these matrices in the following way:
    \begin{itemize}
\item[1)] If $k=l=2$, then
\begin{equation}\label{eqT22}
T_{2,2}=\left(\begin{array}{llllll}I_p&I_p&I_p&\ldots&I_p&P^{-1}\\
I_p&P&P^2&\ldots&P^{p-1}&I_p\\
I_p&P^{2}&P^{4}&\ldots&P^{2(p-1)}&P\\
\multicolumn{6}{l}{\vdots}\\
I_p&P^{p-1}&P^{2(p-1)}&\ldots&P^{(p-1)^2}&P^{p-2}\\
E_0&E_1&E_2&\ldots&E_{p-1}&E_1
\end{array}\right)
\end{equation}
\item[2)] If $k>l$, then
\begin{equation}\label{eqTki}
T_{k,l}=\left(\begin{array}{cc}I_p\otimes T_{k-1,l}&\mbox{\boldmath $0$}\\ \mbox{\boldmath $0$}&I_p\end{array}\right)
\end{equation}
\item[3)] If $k=l>2$ then
\end{itemize}
\begin{equation}\label{eqTkk}
T_{k,k}=\left(\begin{array}{lllll}I_{\theta}\otimes I_p&I_{\theta}\otimes I_p&\ldots&I_{\theta}\otimes I_p&\mbox{\boldmath $1$} \otimes P^{-1}\\
I_{\theta}\otimes I_p&I_{\theta}\otimes P&\ldots&I_{\theta}\otimes P^{p-1}&\mbox{\boldmath $1$} \otimes I_p\\
I_{\theta}\otimes I_p&I_{\theta}\otimes P^{2}&\ldots&I_{\theta}\otimes P^{2(p-1)}&\mbox{\boldmath $1$} \otimes P\\
\vdots\\
I_{\theta}\otimes I_p&I_{\theta}\otimes P^{p-1}&\ldots&I_{\theta}\otimes P^{(p-1)^2}&\mbox{\boldmath $1$} \otimes P^{p-2}\\
E_0\;\;\mbox{\boldmath $0$} &E_1\;\;\mbox{\boldmath $0$} &\ldots&E_{p-1}\;\;\mbox{\boldmath $0$}& I_p
\end{array}\right)
\end{equation}
Here $\otimes$ means Kroneker product and $\theta=\theta(p,k-1)$.
\end{itemize}

\begin{theorem}\label{TM}
Let $\chi$ be a characteristic vector of an $[n,k;q]$-code. Then
%There exist transform matrices $T_{k,i}$, $i=2,\ldots,k$, defined above, such that
\begin{equation}
\left(\widehat{M}_k^{[\chi]}(l)\right)^{\rm T}=T_{k,l}\cdot\left(\widehat{M}_k^{[\chi]}(l-1)\right)^{\rm T}, \quad \;l=2,\dots,k,
\end{equation}
and
\begin{equation}
\left(\widehat{M}_k^{[\chi]}\right)^{\rm T}=T_{k,k}\cdot T_{k,k-1}\cdots T_{k,2}\cdot \widehat{\chi}^{\rm T}
\end{equation}
\end{theorem}

\begin{proof}
Let $k=2$. Then $\theta(p,2)=p+1$, $M_2$ is a $(p+1)\times (p+1)$ matrix, and the characteristic vector $\chi$ has length $p+1$, let $\chi=(\chi_0,\chi_1,\ldots,\chi_p)$.
%Following \eqref{eqChiSplit} we split $\chi$ into $q+1$ integers $\alpha_0,\alpha_1,\ldots,\alpha_q$.
To obtain $M_2^{[\chi]}(2)$, we have to apply the transformations \textsc{Add0}, \textsc{LastRow} and \textsc{AllRows} to $M_2^{[\chi]}(1)=\displaystyle\left(\begin{array}{c} M_1^{[\chi_0]}\\ \vdots\\ M_1^{[\chi_{p-1}]}\\ M_1^{[\chi_p]}
\end{array}\right)$
(see Definition~\ref{defPCChS}). These three transformations have matrix representations. The transform matrices in this case are square matrices of size $q(q+1)$. The three transformation matrices corresponding to \textsc{Add0}, \textsc{LastRow} and \textsc{AllRows}, respectively, are
\[
T_0=\left(
\begin{array}{ccccc}
I_p&\bom{0}&\cdots&\bom{0}&\bom{0}\\
\bom{0}&I_p&\cdots&\bom{0}&P^{-1}\\
&&\ddots&&\\
\bom{0}&\bom{0}&\cdots&I_p&\bom{0}\\
\bom{0}&\bom{0}&\cdots&\bom{0}&I_p
\end{array}
\right), \ \ T_{last}=\left(
\begin{array}{ccccc}
I_p&\bom{0}&\cdots&\bom{0}&\bom{0}\\
\bom{0}&I_p&\cdots&\bom{0}&\bom{0}\\
&&\ddots&&\\
\bom{0}&\bom{0}&\cdots&I_p&\bom{0}\\
E_0&E_1&\cdots&E_{q-1}&\bom{0}
\end{array}
\right),\]
\[ T_{all}=\left(
\begin{array}{cccccc}
I_p&I_p&I_p&\cdots&I_p&\bom{0}\\
I_p&P&P^2&\cdots&P^{p-1}&\bom{0}\\
I_p&P^2&P^4&\cdots&P^{2(p-1)}&\bom{0}\\
\vdots&\vdots&\vdots&\ddots&\vdots&\vdots\\
I_p&P^{p-1}&P^{2(p-1)}&\cdots&P^{(p-1)^2}&\bom{0}\\
\bom{0}&\bom{0}&\bom{0}&\cdots&\bom{0}&I_p
\end{array}
\right).
\]

The matrix $T_{2,2}$ is the product of the above matrices:
\[ T_{2,2}=T_{all}\cdot T_{last}\cdot T_0=\left(\begin{array}{cccccc}I_p&I_p&I_p&\ldots&I_p&P^{-1}\\
I_p&P&P^2&\ldots&P^{p-1}&I_p\\
I_p&P^2&P^4&\ldots&P^{2(p-1)}&P\\
\vdots&\vdots&\vdots&\ddots&\vdots&\vdots\\
I_p&P^{p-1}&P^{2(p-1)}&\ldots&P^{(p-1)^2}&P^{p-2}\\
E_0&E_1&E_2&\ldots&E_{q-1}&E_1
\end{array}\right).
\]

%%%%%%%%%%%%%%%%%%%%%%%%%%%%%%%%%%%%
 Thus $(\widehat{M}_2^{[\chi]})^{\rm T}=T_{2,2}\cdot(\widehat{M}_2^{[\chi]}(1))^{\rm T}$.

Let $k>2$. We assume that the theorem holds for every $k^{\prime}\in\mathbb Z$ where $2\leq k^{\prime}<k$. We split the characteristic vector $\chi\in{\mathbb Z}^{\theta(p,k)}$ into $p+1$ parts according \eqref{eqChiSplit}.

If $k>l$ then
\[M_k^{[\chi]}(l)=\left(\begin{array}{l}M_{k-1}^{[\chi^{(0)}]}(l)\\ M_{k-1}^{[\chi^{(1)}]}(l)\\ \ldots \\ M_{k-1}^{[\chi^{(p-1)}]}(l)\\ M_1^{[\chi^{(p)}]}\end{array}\right)\quad\mbox{and}
\quad M_k^{[\chi]}(l-1)=\left(\begin{array}{l}M_{k-1}^{[\chi^{(0)}]}(l-1)\\ M_{k-1}^{[\chi^{(1)}]}(l-1)\\ \ldots \\ M_{k-1}^{[\chi^{(p-1)}]}(l-1)\\ M_1^{[\chi^{(p)}]}\end{array}\right)\]
Following the induction hypothesis we have
$$(\widehat{M}_{k-1}^{[\chi^{(s)}]}(l))^{\rm T}=T_{k-1,l}\cdot(\widehat{M}_{k-1}^{[\chi^{(s)}]}(l-1))^{\rm T}, \ s=0,1,\ldots,p-1.$$
 So the assertion follows directly.

 If $k=l$ we have
$M_k^{[\chi]}(k)=M_k^{[\chi]}$ and
$M_k^{[\chi]}(k-1)=\left(\begin{array}{l}M_{k-1}^{[\chi^{(0)}]}\\ M_{k-1}^{[\chi^{(1)}]}\\ \ldots \\ M_{k-1}^{[\chi^{(p-1)}]}\\ M_1^{[\chi^{(p)}]}\end{array}\right)$
and we have to apply \eqref{eqChMk}. It turns out that
\[\widehat{M}_k^{[\chi]}=T_{k,k}\cdot\left(\widehat{M}_{k-1}^{[\chi^{(0)}]}| \widehat{M}_{k-1}^{[\chi^{(1)}]}| \ldots | \widehat{M}_{k-1}^{[\chi^{(p-1)}]}| \widehat{M}_1^{[\chi^{(p)}]}\right)^{\rm T}=T_{k,k}\cdot\widehat{M}_k^{[\chi]}(k-1).
\]

The main assertion follows directly.
\end{proof}

\section{Codes over composite fields}
\label{sec:composite}

Let $\F_q=\{0,\alpha_1=1,\alpha_2,\dots,\alpha_{q-1}\}$ be a finite field with $q$ elements, where
$q=p^m$, $p$ is prime and $m>1$. We need a basis
$\beta_1=1$, $\beta_2,\dots,\beta_m$ of $\F_q$ over the prime field $\F_p$.

%\begin{equation}\label{matrices}
%G'=(G | \alpha_2G | \cdots | \alpha_{q-1}G), \ \ \ \ \
%G^*=\left(\begin{array}{c}
%G'\\ \beta_2G'\\ \vdots\\ \beta_mG'
%\end{array}\right).
%\end{equation}
Let $\mbox{Tr}:\F_q\to\F_p$ denote the trace map, and $\mbox{Tr}(a)=(\mbox{Tr}(a_1),\ldots,\mbox{Tr}(a_n))\in\F_p^n$ for $a\in\F_q^n$. Let $C$ be a $[n,k,d]_q$ linear code with a generator matrix $G$, and $G'=(G | \alpha_2G | \cdots | \alpha_{q-1}G)$. The code $\mbox{Tr}(C)=\{\mbox{Tr}(c)\vert c\in C\}$ is the trace code of the linear $q$-qry code $C$. $\mbox{Tr}(C)$ is a linear code over the prime field $\F_p$ with the same length as $C$ but its dimension is less or equal to $mk$ \cite{Trace_Codes}. Therefore instead of $\mbox{Tr}(C)$, we consider the trace code of $C'$, where $C'$ is the code generated by the matrix $G'$ with parameters $[(q-1)n,k,(q-1)d]_q$.

%Then $G^*$ is an $mk\times (q-1)n$ matrix over $\F_q$. Define the map $\mathcal{T}$ which sends any $[n,k]_q$ code $C$ to a code $C_T$ of length $(q-1)n$ over the prime field $\F_p$ by
%$$\mathcal{T}(C)=C_T=\langle (\Tr(g_{ij}))_{mk\times (q-1)n}\rangle, \ \mbox{where} \ G^*=(g_{ij})_{mk\times (q-1)n}.$$
%If $G^*=(g_{ij})_{mk\times (q-1)n}$, the matrix $G_T$ is an $mk\times (q-1)n$ matrix over $\F_p$ whose elements are $Tr(g_{ij})$.

\begin{lemma}\label{dimension}
The dimension of the code $\mbox{Tr}(C')$ is equal to $mk$.
\end{lemma}

\begin{proof}
If $u_1,\ldots,u_k$ and $v_1,\ldots,v_k$ are the rows of $G$ and $G'$, respectively, then $v_i=(u_i | \alpha_2u_i | \cdots | \alpha_{q-1}u_i)$. Let $\beta_1=1,\beta_2,\ldots,\beta_m$ be a basis of $\F_q$ over $\F_p$. We prove that $\mbox{Tr}(\beta_iv_j)$, $i=1,\ldots,m$, $j=1,\ldots,k$, is a basis of the code $\mbox{Tr}(C')$.

Suppose that
$\displaystyle\sum_{i=1}^m\sum_{j=1}^k\lambda_{ij}\mbox{Tr}(\beta_iv_j)=0, \ \lambda_{ij}\in\F_p.$
It turns out that
$$\sum_{i=1}^m\sum_{j=1}^k\lambda_{ij}\mbox{Tr}(\alpha_s\beta_iu_j)= \mbox{Tr}(\sum_{i=1}^m\sum_{j=1}^k\lambda_{ij}\alpha_s\beta_iu_j)=0, \ \forall s\in\{1,2,\ldots,q-1\}.$$
Hence
$$\mbox{Tr}(\alpha_s\sum_{i=1}^m\sum_{j=1}^k\lambda_{ij}\beta_iu_j)=0, \ \forall s\in\{1,2,\ldots,q-1\}.$$

If $\displaystyle\sum_{i=1}^m\sum_{j=1}^k\lambda_{ij}\beta_iu_j\neq 0$ then $\alpha_s(\displaystyle\sum_{i=1}^m\sum_{j=1}^k\lambda_{ij}\beta_iu_j)$, $s=1,2,\dots,q-1$, are all nonzero elements of the field and therefore some of their traces must be nonzero elements of $\F_p$ - a contradiction. This proves that
$\displaystyle\sum_{i=1}^m\sum_{j=1}^k\lambda_{ij}\beta_iu_j= 0$.
Since $u_1,\dots,u_k$ is a basis of the code $C$ then
$$\sum_{i=1}^m
\lambda_{ij}\beta_i= 0, \ \forall j=1,2,\dots,k.$$
But $\beta_1,\beta_2,\dots,\beta_m$ is a basis of $\F_q$ over $\F_p$, so $\lambda_{ij}=0$ for all $i=1,\ldots,m$, $j=1,\ldots,k$. Hence the vectors $\mbox{Tr}(\beta_iv_j)$ are linearly independent and the dimension of $\mbox{Tr}(C')$ is $mk$.
\end{proof}

\begin{corollary}
The codes $C$ and $\mbox{Tr}(C')$ have the same number of codewords, namely $q^k=p^{mk}$.
\end{corollary}

Let $c=(c_1,\ldots,c_n)\in C$ and $c_T=\mbox{Tr}(c\vert \alpha_2c\vert \cdots\vert\alpha_{q-1}c)$. If $c_i\neq 0$ then $\{ c_i,\alpha_2c_i,\dots,\alpha_{q-1}c_i\}=\F_q^*$. Hence $p^m-p^{m-1}$ of the elements in the set $\{ Tr(c_i),Tr(\alpha_2c_i),\dots,Tr(\alpha_{q-1}c_i)\}$ are nonzeros. Hence
$$\wt(c_T)=(p^m-p^{m-1})\wt(c).$$

It turns out that the minimum weight of $\mbox{Tr}(C')$ is $$d_T=(p^m-p^{m-1})d=\frac{q(p-1)}{p}d.$$
So we obtain the following proposition

\begin{proposition}\label{prop1}
If $C$ is an $[n,k,d]$ linear code over $\F_q$, $q=p^m$, $p$ - prime, $m>1$, then $\mbox{Tr}(C')$ is a $[(q-1)n,mk,\frac{q(p-1)}{p}d]_p$ code. Moreover, if $W(y)=\sum_{i=1}^n A_iy^i$ is the weight enumerator of the code $C$, then the weight enumerator of $\mbox{Tr}(C')$ is
$$W_T(y)=\sum_{i=1}^n A_iy^{q(p-1)i/p}.$$
\end{proposition}

%Let focus on the weight distribution of both codes.
%
%\begin{proposition}\label{prop2}
%If the codewords $u$ and $v$ in $C$ are proportional then $u_T$ and $v_T$ have the same weight.
%\end{lemma}
%
%\begin{proposition}
%Let $v=\mu v$, $\mu\in\F_q$. Since the proportional nonzero vectors have the same weight, we have
%$$\wt(v_T)=(p^m-p^{m-1})\wt(v)=(p^m-p^{m-1})\wt(u)=\wt(u_T).$$
%\end{proof}
%
%Both Propositions \ref{prop1} and \ref{prop2} give us that the coefficients in $A_i$ are multiples of $q-1$ and therefore a maximal set of nonproportional codewords of $C$ is enough for computing the weight distribution.

Proposition \ref{prop1} shows that we can use the weight distribution of the code $\mbox{Tr}(C')$ over the prime field $\F_p$ to obtain the weight distribution of the $q$-ary linear code $C$. That's why our algorithm is implemented only for codes over a prime field.

\section{Complexity of the algorithms and experimental results}\label{sec_res}

We consider codes over a prime field $\F_p$ with length $n<2^{32}$ and number of codewords $p^k<2^{64}$, so we need 32-bit integers for the weights of codewords and 64-bit integers for the number of codewords with a given weight.
Therefore we use only basic integer types and operations with them.
To calculate the weight distribution of a linear code, we use two arrays with 32-bit integers, namely $H$ of size $\theta(p,k)\times p$ and $T$ of size $p\times p$. The total memory we need (without a memory for the generator matrix) is $p\theta(p,k)+p^2+2n+C$ 32-bit units, where we add $2n$,  because the weight distribution is a vector of length $n$ consisting of 64-bit integers, and a constant $C$ for the other variables in the algorithms. If we use the reduced weighted distribution, we will have one column less in the array $H$, so we have to subtract $\theta(p,k)$ from the above expression.
%The characteristic vector of the code and the weight distribution are integer vectors of lengths $\theta(q,k)$ and $n$, respectively.

%To compute the reduced weighted distribution we use a $\theta(q,k)\times (q-1)$ array $M_k^{[\chi]_r}(l)$, $k$-elements array $a$ for the inactive rows, $(q+1)\times (q-1)$ help array for calculating a new tuple of rows and a few variables.
%In a fixed step we have $\theta(q,k-l)$ inactive rows and $\theta(q,k)-\theta(q,k-l)=q^{k-l}\theta(q,l)$ rows to be updated. We give more calculations on the complexity in Section \ref{sec_res}.
%To compute a new row we have no more $(q-1)(2q)=O(q^2)$ operations for an intermediate row and $q\theta(q,l-1)$ operations for the last row of a tuple. So the summary complexity of the considered algorithm is $O(k\cdot q^{k+1})$.

The main procedure computes the array $H$ in $k-1$ steps. In the $l$-th step of the procedure, there are  $a_l$ active and $b_l$ inactive rows, where $a_l+b_l=\theta(p,k)$, $a_l=p^{k-l}\theta(p,l)$, $b_l=\theta(p,k-l)$, $l=2,3,\dots,k$. The inactive rows remain unchanged. Any element in an active row is calculated in Algorithm \ref{alBChS} as a sum with $p$ summands. There are $a_l$ active rows of length $p$ and so we use $a_lp^2$ operations for the calculations in this step. Actually, this is the number of calculations of the transformations \textsc{LastRow} and \textsc{AllRows}. The transformation \textsc{Add0} uses $p^{k-l}\theta(p,l-1)\le \theta(p,k-1)$ operations, and therefore the complexity of the $l$-th step (the body of the for-loop) is
\begin{align*}
  a_lp^2+p^{k-l}\theta(p,l-1) & =p^{k+2-l}\frac{p^l-1}{p-1}+p^{k-l}\frac{p^{l-1}-1}{p-1} \\
   & =\frac{p^{k+2}-p^{k+2-l}+p^{k-1}-p^{k-l}}{p-1}.
\end{align*}
Hence the complexity of Algorithm \ref{alCChS} is
$$\sum_{l=2}^k \frac{p^{k+2}-p^{k+2-l}+p^{k-1}-p^{k-l}}{p-1}=(k-1)\frac{p^{k+2}+p^{k-1}}{p-1}-\frac{(p^2+1)(p^{k-1}-1)}{(p-1)^2}.$$
It turns out that for a fixed $p$ the complexity of the algorithm is $O(kp^k)$. When accounting for both $k$ and $p$, in terms of arithmetic
operations the running time can be written as $O(kp^{k+1})$.

\begin{remark}
We compare our algorithm with Algorithm 9.8 (Walsh transform over a prime finite field $\F_p$) in \cite{Joux}.  According to Joux, the complexity of his algorithm when $p$ varies is $O(kp^{k+2}$).
%The advantage of our algorithm is that we use only nonproportional vectors.
\end{remark}

%Since Algorithm \ref{alCChS} is presented in iterative form, its complexity analysis is
%not difficult. We use a $\theta(q,k)\times q$ array $H$, a $k$-elements array $a$, a $q\times q$ help array $TEMP$ and a few variables. We have to provide $k-1$ steps. In a fixed step we have $\theta(q,k-l)$ inactive rows and $\theta(q,k)-\theta(q,k-l)=q^{k-l}\theta(q,l)$ rows to be updated. To compute a new row we have $O(q^2)$ operations. So the summary complexity of the considered algorithm is $O(k. \theta(q,k). q^2)=O(k. q^{k+1})$ if $q$ varies and $O(k.q^k)$ for a fixed $q$.

We implement the presented approach, based on Algorithms~1--3, in a C/C++ program. To compare the efficiency,  we use C implementation of an algorithm, presented in  \cite{BouyuklievBakoev}, with the same efficiency as the Gray code algorithms.  As a development environment for both algorithms we use \textsc{MS Visual Studio 2012}.  All examples are executed on  (\textsc{Intel Core  i7-3770k 3.50 GHz processor}) in Active solution configuration --- Release, and Active solution platform --- \textsc{X64}.

Input data are randomly generated linear codes with lengths 30, 300, 3000, 30000 and different dimensions over finite fields with 2, 3, 4, 5, and 7 elements.
All the results with the obtained execution times are given in seconds (Table 1). Any column consists of two subcolumns. The first subcolumn (named 'NEW') contains the results obtained by the new algorithm (described in this paper), and the second one gives the execution time for the same code but using the algorithm from \cite{BouyuklievBakoev}, implemented in the package \textsc{Q-Extension}. The runtime shown in Table 1 is the full execution time to compute the weight distribution starting with a generator matrix of a code with the given parameters.

In Table 2 we present results for the same parameters as in Table 1 but obtained  using Magma V2.25-2 by online Magma Calculator run in a virtual machine on an Intel Xeon Processor E3-1220, 3.10 GHz.

%on a Linux system with processor Intel(R) Core(TM) i5-4570 CPU @ 3.20GHz (averaged over 5 runs).

The results given in the tables show that the  presented approach is faster for codes with large length. The execution time for computing the characteristic vector is  negligible.

%In conclusion, we can say that this approach is  very fast, easy for  parallelization, but it needs a lot of memory.

%\begin{table}
%\caption{Experimental results}
%\label{table-1}       % Give a unique label
%\begin{tabular}{c|ccccccc}
%\hline\noalign{\smallskip}
%$n\setminus k$ & 2 & 3&4&5&6&7&8  \\
%                                                                                                                                                                                                                                                                                                                                                                                                                                                                                                                        \noalign{\smallskip}\hline\noalign{\smallskip}
%total &  45& 564  &14219  & 97897 & 46776 & 38507 & 27299  \\                                                                                                                                                                                                                                                                                                                                                                                                                                                                                                                                                \noalign{\smallskip}\hline
%\end{tabular}
%\end{table}

\begin{table}
\caption{Experimental results}
\label{table-1}       % Give a unique label

\begin{tabular}{c|c||rr|rr|rr|rr}
\hline\noalign{\smallskip}
 &&  $n=$& 30 & $n=$&300   & $n=$&3000    &$n=$&30000 \\
\noalign{\smallskip}\hline\noalign{\smallskip}
$q$&$k$&NEW  &OLD & NEW   &  OLD      &NEW   &  OLD      & NEW  &OLD \\
\noalign{\smallskip}\hline\noalign{\smallskip}
2 & 23 & 0.190 &0.084 & 0.132 & 0.139  & 0.131  & 0.700   & 0.133 &  6.064  \\
2 & 24 & 0.214 &0.168 & 0.268 & 0.278  & 0.266  & 1.412   & 0.271  & 11.665\\
2 & 25 & 0.552 &0.337 & 0.552 & 0.549  & 0.552  & 2.816   & 0.552  &23.649 \\
2 & 26 & 1.146 &0.637 & 1.144 & 1.107  & 1.148   & 5.595  & 1.150  & 47.001 \\
\hline
3 & 13 &0.039 &0.015 & 0.101   & 0.032     & 0.101  & 0.190       &  0.103 & 1.690 \\
3 & 14 &0.292 &0.048 & 0.295   & 0.099      &0.294   & 0.575       & 0.295  &5.070  \\
3 & 15 &0.968 &0.145 & 0.949    & 0.291      & 0.955  & 1.716     & 0.957  & 15.122 \\
3 & 16 &3.012 &0.439 & 3.035   &  0.881     &  3.100  & 5.249       & 3.111  & 46.695 \\
\hline
4 &10  &0.016 &0.007&  0.016   & 0.013 & 0.016  & 0.089 & 0.016  &0.747 \\
4 &11  &0.064 &0.025 & 0.064   & 0.052 & 0.065  & 0.344 & 0.066  &2.997  \\
4 &12  &0.261 &0.109 & 0.263   & 0.208 & 0.263  & 1.325 & 0.263  &11.637  \\
4 & 13 &1.444 &0.422 & 1.444   & 0.857 & 1.445  & 5.359 & 1.446  & 46.976 \\
\hline
5 & 8 &0.130 &0.013 & 0.140   & 0.133   & 0.140 &    1.225    & 0.150  &12.228  \\
5 & 9 &0.063 &0.068 & 0.063   &  0.614     &0.065   & 6.135       & 0.067  & 60.834 \\
5 & 10 &0.335 &0.309 & 0.335   & 3.062      & 0.337  & 30.318       & 0.343   & 295.691  \\
5 & 11 &1.847 &1.517 & 1.842   & 15.197      &  1.841 &   151.716     & 1.843  & 1514.924 \\
\hline
7 &6  &0.004 &0.002 & 0.004  & 0.020      & 0.004  & 0.202      & 0.008  & 2.049 \\
7 &7  &0.027 &0.013 & 0.022  &  0.166     & 0.021  & 1.469       & 0.026  &14.554  \\
7 &8  &0.170 &0.107 & 0.174   & 1.037      & 0.172  & 10.084       & 0.181  & 101.280 \\
7 &9  &1.351 &0.743 & 1.363   & 7.105     & 1.379  &  70.795      & 1.397  &705.218  \\

\noalign{\smallskip}\hline
\end{tabular}
\end{table}

\begin{table}
\caption{Experimental results using Magma V2.25-2 by online Magma Calculator}
\label{table-Magma}
%\footnotesize
%\begin{center}
%Table 1
%\caption{RES.\label{SO3}}\vspace*{0.2in}
%{\def\arraystretch{1}%\small
\begin{tabular}{r|r||r|r|r|r}
%\hline $$&& & n= 30 & n=300 &   n=3000 &  n=30000\\
\hline\noalign{\smallskip}
% &  $d$ & \# &  $d$ & \# & $d$ & \# & $d$ & \#\\
q&k& n= 30   &  n= 300     & n= 3000        &  n= 30000  \\\hline\hline
2 & 23  &0.000 & 0.130 & 1.140  & 11.340 \\
2 & 24  &0.000 & 0.260 & 2.300  & 22.680\\
2 & 25  &0.000 & 0.520 & 4.560 & 45.380\\
2 & 26  &0.000 & 1.030 & 9.150 & 90.680 \\
\hline
3 & 13  &0.010 & 0.040 & 0.180  & 1.730 \\
3 & 14   &0.010 & 0.080 & 0.540 & 5.180  \\
3 & 15  &0.060 & 0.240 &  1.580 & 15.530 \\
3 & 16  &0.020 & 0.630 & 4.810 & 46.580 \\
\hline
4 &10   &0.010 & 0.040 & 0.090  &0.730 \\
4 &11   &0.010 & 0.070 & 0.320 & 2.940  \\
4 &12   &0.060 & 0.180 & 1.250 & 11.760  \\
4 & 13  &0.230 & 0.680 & 5.000 & 47.030\\
\hline
5 & 8    &0.00 & 0.070 & 0.040  & 0.330 \\
5 & 9    &0.00 & 0.090 & 0.170 &   1.670 \\
5 & 10  &0.03 & 0.170 & 0.860 &   8.350  \\
5 & 11  &0.160 & 0.600 & 4.270 & 41.770 \\
\hline
7 &6  &0.00 & 0.00 &   0.010 &   0.080 \\
7 &7   &0.00 & 0.010 &   0.060 &   0.530 \\
7 &8   &0.010 & 0.140 &   0.370 &   3.600 \\
7 &9   &0.100 & 0.410 &   2.630 & 25.300 \\
\noalign{\smallskip}\hline
\end{tabular}
\end{table}

\section*{Acknowledgements}

This research was supported by Grant DN 02/2/13.12.2016 of the Bulgarian National Science Fund. The third author was partially supported by JSPS KAKENHI Grant Number JP16K05256.

We thank Geoff Bailey, Computational Algebra Group, University of Sydney, for the provided information about the processor used by Magma Calculator.

%We are greatly indebted to the unknown referees for their careful reading of the manuscript and for their
%useful suggestions.



\end{document}